%% file: main.tex
\documentclass[conference]{IEEEtran}
\IEEEoverridecommandlockouts
\usepackage{cite}
\usepackage{amsmath,amssymb,amsfonts}
\usepackage{amsthm}
\usepackage{graphicx}
\usepackage{textcomp}
\usepackage{xcolor}
\usepackage{algorithm}
\usepackage{algpseudocode}
\usepackage{array}
\usepackage{multirow}
\usepackage{caption}
\usepackage{comment}
\usepackage{subcaption}
\usepackage{multicol}
\usepackage{graphicx}
\usepackage{asymptote}
\usepackage{array}
\renewcommand{\baselinestretch}{0.98}
\newtheorem{lemma}{Lemma}
\newtheorem{remark}{Remark}

\def\BibTeX{{\rm B\kern-.05em{\sc i\kern-.025em b}\kern-.08em
    T\kern-.1667em\lower.7ex\hbox{E}\kern-.125emX}}

\newtheorem{theorem}{Theorem}
\begin{document}

\title{Convexity and Optimization in Deficit Round Robin Scheduling for Delay-Constrained Systems\\
}
\author{Aniket~Mukherjee, Joy~Kuri and~Chandramani~Singh
}

\maketitle
\thispagestyle{plain}
\pagestyle{plain}
\input{1abstract}
\input{2keywords}
\input{3introduction}
\input{7networkslicingsmandb}

\input{8modifiedbound}
\input{9networkslicingfinding}

\input{10Minimalswitching2-flow}
\input{11nflow-case}

\input{12results}

\input{13conclusion}

\bibliography{Ref}
\bibliographystyle{ieeetr}
\input{14Appendix}

\end{document}

%% file: 1abstract.tex
\begin{abstract}
The Deficit Round Robin (DRR) scheduler is widely used in network systems for its simplicity and fairness. However, configuring its integer-valued parameters, known as quanta, to meet stringent delay constraints remains a significant challenge. This paper addresses this issue by demonstrating the convexity of the feasible parameter set for a two-flow DRR system under delay constraints. The analysis is then extended to n-flow systems, uncovering key structural properties that guide parameter selection. Additionally, we propose an optimization method to maximize the number of packets served in a round while satisfying delay constraints. The effectiveness of this approach is validated through numerical simulations, providing a practical framework for enhancing DRR scheduling. These findings offer valuable insights into resource allocation strategies for maintaining Quality of Service (QoS) standards in network slicing environments.
\end{abstract}

%% file: 2keywords.tex

%% file: 3introduction.tex
\section{Introduction and Related Work}
Deficit Round Robin (DRR) \cite{Shree-varghese-drr-96, Lenzini-Adrr-02} is a popular scheduling algorithm known for its simplicity and fairness in handling multi-flow systems. It works by assigning a quantum to each flow, which determines the amount of data served in each scheduling round. This mechanism allows proportional allocation of service capacity based on flow-specific requirements. Despite its advantages, configuring DRR—particularly determining the optimal quantum values—remains challenging, especially in systems with diverse traffic patterns and strict delay requirements.

In this paper, we study a single-node communication system with $n$ incoming data flows served by a constant-rate server of capacity $c$. Each flow is modeled using a bursty arrival envelope (leaky bucket) and is associated with specific delay requirements. Using a DRR scheduler, the goal is to allocate service efficiently while meeting these delay constraints, maintaining resource utilization, and minimizing the number of context switches between flows. This balance becomes increasingly difficult as the number of flows grow, underscoring the need for systematic quanta configuration.

Recent studies have focused on delay bounds and worst-case delay analysis for DRR schedulers \cite{Boyer_2012, SEYED_ET_ALL_DRR_22, Bouillard_bandwidth-sharing_21, Hua_DRR_bound, Soni_DRR_bound_2018} using Network Calculus \cite{NC_GUIDE} and their application to modern technologies such as network slicing and time-sensitive networking \cite{Coronado-Lasagna-2018, Ali-DRR-2021}. While the approach Fuzzy Adaptive DRR (FADRR) \cite{Ali-fuzzyDRR-2014} dynamically adjusts quantum values to meet QoS targets it lacks the precision required for strict delay guarantees. \cite{soni_2019, Soni_2020} propose binary search-based heuristic algorithms to optimize bandwidth allocation. Our work builds on existing delay-bound analyses, particularly the recent bounds derived in \cite{SEYED_ET_ALL_DRR_22}, to address the following fundamental questions:
\begin{itemize}
    \item \textbf{Existence:} Can we find a vector of quanta that satisfies all delay requirements for multiple flows?
    \item \textbf{Optimization:} If such a vector exists, how can we determine the optimal quantum values that balance delay guarantees, minimize switching overhead, and maximize resource utilization?
\end{itemize}

Using a modified delay bound based on \cite{SEYED_ET_ALL_DRR_22}, we prove that the feasibility set of quanta is convex. This insight simplifies the optimization process, allowing for efficient computation of optimal quanta that minimize scheduling inefficiencies. By focusing on reducing the number of context switches while ensuring delay constraints are met, we strike a balance between system performance and fairness.
The paper is organized as follows: Section \ref{SM} introduces the system model and provides the foundation for the discussions that follow. Section \ref{MB} presents the delay bound, the modified delay bound, and formulates the optimization problem addressed in this paper. Section \ref{C2F} establishes the properties of the optimization problem's constraint set, including conditions for non-emptiness and convexity. Section \ref{OP2F} proposes two algorithms to determine the optimal quanta values that meet the deadline and proves that the algorithms converge to the optimal solution and presents the simulation results, and Section \ref{conclusion} summarizes the key contributions and findings.

%% file: 7networkslicingsmandb.tex
\section{System Model}\label{SM}
In this section, we present the system model for the problem of finding the optimal quanta of DRR scheduler to meet the delay requirements and also minimize the switching. Our focus is on a single-node system that operates with a constant-rate server, supporting multiple incoming data flows. Each flow exhibits bursty traffic characteristics, which are managed by DRR scheduling algorithm. We describe the system in detail, including the traffic model, the scheduling algorithm, and the key assumptions made to facilitate analysis.
\subsection{System Setup}
We consider a communication system where a constant-rate server, with service rate $c$, serves $n$ incoming data flows. Each flow is indexed by $i \in {1, 2, \dots, n}$. The data traffic for each flow is modeled using a bursty rate arrival envelope, commonly referred to as a leaky bucket model. This model allows us to capture both the burstiness and the rate-limited nature of traffic. The bursty behavior of flow $i$ is parameterized by two parameters: the burst size $b_i$, which is the maximum burst size that the flow can generate, and the arrival rate $r_i$, which specifies the rate at which data arrives for this flow over time. The total data arrival for flow $i$ over the time interval $(0, t)$ is thus bounded above by the following affine function: $b_i+r_i t$. The objective is to efficiently allocate the server's resources by determining the parameter of DRR scheduler for each flow while ensuring that their respective delay constraints are satisfied.

\subsection{Scheduling: Deficit Round Robin (DRR)}

Deficit Round Robin (DRR) is a scheduling algorithm commonly used in real-time systems and communication networks to allocate service among multiple queues or flows. Each flow $i$ is assigned a fixed quantum $q_i$, which determines the amount of data it can transmit in each round of scheduling. The system maintains one queue per input and runs an infinite loop of scheduling rounds. In each round, if a queue is non-empty, its service begins, and its deficit is increased by $q_i$. The scheduler serves packets from the queue until the deficit is smaller than the size of the head-of-the-line packet or the queue becomes empty. If the queue becomes empty, the deficit is reset to zero; otherwise, the remaining deficit (residual deficit) is carried over to the next round. This ensures that flows receive service proportional to their assigned quantum. 

\begin{figure}[htbp!]
    \centering
    \includegraphics[width=0.5\textwidth]{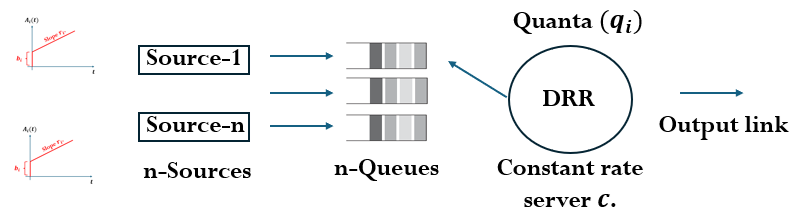}
    \caption{The system as defined has $n$-sources, that are bounded by leaky bucket sources (parameterized with $r_i, b_i$), and a DRR-scheduler with constant rate $c$, such that the quanta are $q_i^{th}$ for $i^{th}$ queue.}
    \label{system_pic}
\end{figure}

\subsection{Delay Requirements and the Impact of Quanta}
Each flow in the system has a specific delay requirement denoted as $d_i$. The delay requirement represents the maximum permissible delay for the data flow to meet quality-of-service (QoS) constraints. We analyze how the quanta $q_i$ assigned to each flow can affect the delay requirements and whether we can satisfy these requirements for all flows.

\begin{table}[htbp]
\caption{System Model}
\begin{center}
\begin{tabular}{|c|p{5.5cm}|}
\hline
\multicolumn{2}{|c|}{\textbf{Input Parameters}} \\
\hline
$b_i$ & Burst $i^{th}$ flow\\
\hline
$r_i$ & Arrival rate $i^{th}$ flow\\
\hline
$d_i$ & Delay requirement $i^{th}$ flow\\
\hline
$L$ & Max residual deficit\\
\hline
$c$ & Capacity of the server\\
\hline
$D_i$ & Upper bound on the delay of $i^{th}$ flow as per \cite{SEYED_ET_ALL_DRR_22}\\
\hline
$\hat{D}_i$ & Modified upper bound\\
\hline
\multicolumn{2}{|c|}{\textbf{Variable}} \\
				\hline
$q_i$ & Quanta (DRR-parameter) for $i^{th}$ flow\\
\hline
\end{tabular}
\label{TSN_system_model}
\end{center}
\end{table}

Therefore, to facilitate the analysis and make the problem more tractable, we introduce an modified upper bound on the delay, which is simpler to analyze. The introduction of this upper bound reveals that the set of quanta satisfying the delay requirements form a convex set, which greatly simplifies the search for an optimal solution. All the parameters and variables are stated in Table~\ref{TSN_system_model}








%% file: 8modifiedbound.tex
\section{Delay Bound analysis and optimization problem}\label{MB}
In this section we define the delay bound of $i^{th}$ as presented in \cite{SEYED_ET_ALL_DRR_22} and adapt it to our system. The delay bound for flow $i$ is given by the following expression:

\begin{align}\label{seyed_bound}
    D_i = \max\left(\frac{\Psi_i(b_i)}{c}, \frac{\Psi_i(\alpha_i(\tau_i))}{c} - \tau_i\right),
\end{align}
where:
\begin{align}
    \tau_i &= \frac{q_i - \left(b_i + L\right) \bmod q_i}{r_i},\nonumber\\
    \Psi_i(x) &= x + \sum_{j \neq i} \phi_{i,j}(x), \label{psi} \\
    \phi_{i,j}(x) &= \left\lfloor \frac{x + L}{q_i} \right\rfloor q_j + (q_j + L). \label{phi}
\end{align}

In the above equations, we define several important components:
$D_i$ represents the delay bound for flow $i$, which is the maximum delay the flow can experience.
$\Psi_i(x)$ represents the total service demand for flow $i$, which includes the data generated by flow $i$ itself as well as the contributions from other interfering flows. The term $\phi_{i,j}(x)$ is the contribution from flow $j$ to the service demand of flow $i$, and it is computed based on the quantum assigned to flow $i$ and the interference from other flows.
$\alpha_i(\tau_i)$ is the arrival curve for flow $i$, which specifies the maximum number of bits that can arrive at flow $i$ by time $\tau_i$.
    
We propose a modified delay bound, $\hat{D}_i$. This modification simplifies the analysis while ensuring that the original delay bound $D_i$ is always less than or equal to $\hat{D}_i$.
\begin{lemma}[Modified Delay Bound]\label{modifieddelaybound}
    Let $D_i$ denote the delay bound as defined in \cite{SEYED_ET_ALL_DRR_22}. Then, the modified delay bound $\hat{D}_i$ satisfies $D_i \leq \hat{D}_i$, where:
    \begin{align}
        \hat{D}_i &= \frac{b_i + L}{c} \left(1 + \frac{\sum_{j \neq i} q_j}{q_i}\right) 
        + \frac{\sum_{j \neq i} q_j}{c} + \frac{(n-2)L}{c}\nonumber \\
        &\quad + \left(q_i \left(\frac{r_i - c}{r_i c}\right) + \frac{\sum_{j \neq i} q_j}{c}\right)^+.\label{modifiedboundreal}
    \end{align}
\end{lemma}
\begin{proof}
    See Appendix~\ref{app:lemma_proof}
\end{proof}

The modified delay bound in Lemma~\ref{modifieddelaybound} is met with equality when $(b_i + L) \bmod q_i = 0$. 
To simplify the analysis, we assume that the quanta $q_i$ are real-valued, rather than integer-valued. This assumption simplifies the problem and enables us to leverage the structure of the modified upper bound, making it easier to investigate the convexity of the parameter set that satisfies the delay requirements, specifically ensuring that $\hat{D}_i \leq d_i$ for each flow.

It is important to note that the set of $(q_1, q_2)$ that satisfy the original delay constraints, \( D_i \leq d_i \), is not convex. To demonstrate this, consider the following parameter values:  $b_1 = b_2 = 10, r_1 = r_2 = 1, d_1 = d_2 = 1, c = 40, L = 3$. We compute the delay bounds for both flows using \eqref{seyed_bound} at the points $(q_1, q_2) = (5, 9), (6,10)$ and $(7,11)$. The corresponding delay bounds are:  
$(D_1, D_2) = (1, 0.575), (1.075, 0.625)$ and  $(0.875, 0.675).$ From these values, we observe that the points $ (5,9) $ and $(7,11)$ satisfy the delay constraints, whereas the point $(6,10)$ does not. This confirms that the set of $(q_1, q_2)$ that satisfy the delay requirement $D_i \leq d_i$ is not convex.

Introducing the modified bound $\hat{D}_i$ is important because it makes the problem easier to handle. This change helps us to analyze the shape of the feasible set, which is a key step for solving the optimization problem.

From Lemma~\ref{modifieddelaybound}, we know that $D_i \leq \hat{D}_i \leq d_i$. This means that if the modified upper bound $\hat{D}_i$ satisfies the delay requirement $d_i$, then the original delay bound $D_i$ will also be met. So, ensuring $\hat{D}_i \leq d_i$ is enough to guarantee that the delay requirement is satisfied.

We show that the set of quanta $q = (q_1, \dots, q_n)$ that satisfy the modified delay bounds for both flows, i.e., $\hat{D}_i \leq d_i$ for all $i \in \{1, \dots, n\}$ forms a convex set. To do this, we define function $f^i(q)$ that represent the deviations from the delay requirements for each flow i.e., $\{q > 0 |\hat{D}_i \leq d_i\} = \{q > 0| f^i(q) \leq 0\}$. 

The function $f^i(q)$, which represents the deviation for flow $i$, is given by:
\begin{align}
f^i(q) &= \frac{b_i + L}{c} \left(1 + \frac{\sum_{j \neq i} q_j}{q_i}\right) - d_i + \frac{\sum_{j \neq i} q_j}{c} 
    + \frac{(n-2)L}{c} \nonumber\\ & + \left(q_i \left(\frac{r_i - c}{r_i c}\right) + \frac{\sum_{j \neq i} q_j}{c}\right)^+. \label{deviation_n_flow}
\end{align}
where the terms are defined as before.
We express the sets of quanta $q$ that satisfy the delay constraints for each flow as:
\begin{align}\label{2-flow-fesability-set}
\mathcal{D}^i &= \{q > 0 \mid f^i(q) \leq 0\}
\end{align}

To maximize the number of packets served per round, our goal is to maximize the total quanta while ensuring that the delay constraints for each of the $n$ flows are satisfied. This optimization problem is formulated as follows:

\begin{equation}\label{n-flow-opti}
\begin{aligned}
   \text{maximize} \quad & \sum_{i=1}^{n} q_i \\
   \text{subject to} \quad 
        & q\in \bigcap_{j=1}^{n} \mathcal{D}^j,
\end{aligned}
\end{equation}

where the feasible region is defined by the intersection of the constraint sets $\mathcal{D}^j$ for each flow $j=1,2,\dots,n$. These constraint sets determine the permissible values for the quanta of each flow, ensuring that the system adheres to the required delay constraints. The objective of this optimization problem is to maximize the total sum of all quanta, thereby maximizing the number of packets that can be served in each scheduling round while maintaining system efficiency and fairness.

%% file: 9networkslicingfinding.tex
\section{Properties of constraint set}\label{C2F}
In this section, we establish necessary conditions to ensure that the feasible set of the optimization problem defined in \eqref{n-flow-opti} is non-empty we also demonstrate that the feasible set is convex. First, we derive a necessary condition for the existence of a feasible solution in the case of $n$ flows. To prove convexity, we begin by analyzing the two-flow case, demonstrating that its feasible set is convex. We then extend this result to the general $n$-flow case, proving that the feasible set remains convex for any number of flows.

\subsection{Necessary Conditions for a Non-Empty Feasible Set}

In this subsection, we establish the necessary conditions for the feasible set to be non-empty. The feasible set is given by  

\begin{equation}
    \bigcap_{i=1}^n \{ q > 0 \mid f^i(q) \leq 0 \}.
\end{equation}

From \eqref{deviation_n_flow}, we observe that the condition $f^i(q) \leq 0$ implies  

\begin{align}\label{implication_non_empty_n}
    \frac{b_i + L}{c} \left( 1 + \frac{\sum_{j \neq i} q_j}{q_i} \right) 
    - d_i + \frac{\sum_{j \neq i} q_j}{c} + \frac{(n-2)L}{c} \leq 0.
\end{align}

Using the implication in \eqref{implication_non_empty_n}, we derive the necessary condition for the non-emptiness of the feasible set.

\begin{lemma}\label{n_flow_non_empty_lemma}
    If the set $\bigcap_{j=1}^n \mathcal{D}^j$ is non-empty, then  
    \begin{equation}
        c \geq \sum_{i} \frac{b_i + L}{q_i}.
    \end{equation}
\end{lemma}

\begin{proof}
    Let $q \in \bigcap_{j=1}^n \mathcal{D}^j$, which implies that for any $i \in \{1, \dots, n\}$,
    \begin{align}
        \sum_{j \neq i} q_j &\leq \frac{q_i (c d_i - b_i - (n-1)L)}{b_i + L + q_i}.
    \end{align}
    
    Furthermore, we have
    \begin{align}
        \frac{\sum_{j} q_j}{q_i} + \frac{\sum_{j \neq i} q_j}{b_i + L} 
        &\leq \frac{c d_i}{b_i + L} - \frac{b_i + nL}{b_i + L}.
    \end{align}
    
    Since $\frac{\sum_{j\neq i} q_j}{b_i + L} > 0$ and $\frac{b_1 + nL}{b_i + L} > 0$, we can rewrite the above inequality as  
    \begin{align}
        \frac{b_i + L}{c d_i} \leq \frac{q_i}{\sum_j q_j}.
    \end{align}

    We have $n$ such inequalities for each $i \in \{1, \dots, n\}$. Summing all these inequalities, we obtain  
    \begin{align}
        \sum_{i} \frac{b_i + L}{c d_i} \leq c.
    \end{align}
\end{proof}

The implication of Lemma~\ref{n_flow_non_empty_lemma} is that it establishes the minimum bandwidth necessary to meet the system’s deadline constraints. Specifically, the total capacity must exceed the sum of the clearing times necessary to process the bursts of all flows. This ensures that the system has sufficient resources to handle the maximum load without violating timing constraints. By providing a clear criterion for bandwidth allocation, the lemma highlights the critical relationship between flow characteristics and system performance.

\subsection{Convexity of feasible set: Two flow case}
In this subsection, we demonstrate that the feasible set $\mathcal{D}^1 \cap \mathcal{D}^2$ is convex. From the definition of $f^i(q)$ in \eqref{deviation_n_flow}, we can express the same function for the two-flow case, specifically for flow 1:

\begin{align}\label{2-flow-deviations}
    f^1(q_1, q_2) &= \frac{b_1 + L}{c} \left(1 + \frac{q_1 + q_2}{q_1}\right) - d_1 + \frac{q_2}{c} 
     \nonumber\\ & + \left(q_1 \left(\frac{r_1 - c}{r_1 c}\right) + \frac{q_2}{c}\right)^+.
\end{align}

To facilitate our analysis, we introduce the following auxiliary functions that decompose different components of \eqref{2-flow-deviations}:

\begin{itemize}
    \item \textbf{Linear constraint term}:  
    \begin{align}
        a(q_1, q_2) &= \frac{q_1 (r_1 - c)}{r_1 c} + \frac{q_2}{c}, \label{affine-constraint-part}
    \end{align}
    This term appears inside the max function $(\cdot)^+$, determining when the second component in $f^1(q_1, q_2)$ is active.

    \item \textbf{When linear term is inactive ($a(q_1, q_2) \leq 0$)}:  
    \begin{align}
        f_1^1(q_1, q_2) &= \frac{b_1 + L}{c} \left( 1 + \frac{q_2}{q_1} \right) + \frac{q_2}{c} - d_1, \label{if-affine-negative}
    \end{align}
    This function represents the deviation when the affine term in \eqref{2-flow-deviations} does not contribute, i.e., when $a(q_1, q_2) \leq 0$. Additionally, the feasibility condition associated with this case is governed by:
    \begin{align}
        h_1^1(q_1) &= \frac{q_1(cd_1 - b_1 - L)}{(b_1 + L) + q_1}, \label{h_11}
    \end{align}
    This function represents an upper bound on $q_2$ when $f^1(q_1, q_2) \leq 0$.

    \item \textbf{When linear term is active ($a(q_1, q_2) > 0$)}:  
    \begin{align}
        f_2^1(q_1, q_2) &= \frac{b_1 + L}{c} \left( 1 + \frac{q_2}{q_1} \right) + \frac{2q_2}{c} - d_1 + \nonumber\\& q_1\left( \frac{r_1 - c}{r_1 c} \right), \label{if-affine-positive}
    \end{align}
    This function represents the deviation when the affine term in \eqref{2-flow-deviations} contributes, i.e., when $a(q_1, q_2) > 0$. The associated feasibility condition is given by:
    \begin{align}
        h_2^1(q_1) &= \frac{q_1^2 (c - r_1) + q_1 r_1 (c d_1 - b_1 - L)}{r_1 (b_1 + L) + 2 r_1 q_1}, \label{h_21}
    \end{align}
    This function represents an upper bound on $q_2$ when the max term in $f^1$ is active.
\end{itemize}

The term $cd_1$ represents the total work completed by the server in the interval $(0,d_1)$, assuming a service rate of $c$. Meanwhile, $b_1+L$ represents the work that can arrive instantaneously for the $i^{th}$ flow. To ensure the delay target is met, the server must process more work than what has arrived. This requires $cd_1 \geq b_1+L$, meaning that the delay target is only achievable if the server has enough time to handle the initial worst-case burst. If this condition is not met, the burst cannot be fully processed within the required delay.
By analyzing these functions, we can systematically establish the convexity of the feasible set $\mathcal{D}^1 \cap \mathcal{D}^2$.





We are interested in the inequalities $f_1^1(q_1, q_2) \leq 0$ and $f_2^1(q_1, q_2) \leq 0$, which can be written as:

\begin{align}
    q_2 &\leq h_1^1(q_1), \\
    q_2 &\leq h_2^1(q_1).
\end{align}

The functions $h_1^1(\cdot)$ and $h_2^1(\cdot)$ are defined in \eqref{h_11}, \eqref{h_21}. These inequalities describe the feasible region for $q_1$ and $q_2$.



We rewrite the set $\mathcal{D}^1$ based on the sign of  $a(q_1, q_2)$. Since $a(q_1, q_2)$ is a linear, it partitions the space into two disjoint half spaces.
Let us denote the half-space where $a(q_1, q_2) \leq 0$ as the set $A$:
\begin{align}
    A = \left\{(q_1, q_2) \in \mathbb{R}^2_+ \mid a(q_1, q_2) \leq 0 \right\}.    
\end{align}

From the definition, we know that both $A$ and $A^c$ are convex sets. 

Let the sets $B_1$ and $B_2$ be defined as:
\begin{align}
    B_1 &= \left\{ f_1^1(q_1, q_2) \leq 0\right\} = \left\{ q_2 \leq h_1^1(q_1)\right\}, \\
    B_2 &= \left\{ f_2^1(q_1, q_2) \leq 0 \right\} = \left\{ q_2 \leq h_2^1(q_1) \right\}.
\end{align}

Therefore, the set $\mathcal{D}^1$ can be equivalently written as
\begin{align}
    \mathcal{D}^1 &= \left(A \cap B_1 \right) \cup \left(A^c \cap B_2 \right),
\end{align}

The sets $A$, $B_1$, and $B_2$ are all shown in the Figure~\ref{A,B}. Before proving that $\mathcal{D}^1$ is a convex set, we will establish the following properties:

\begin{figure}[htbp!]
    \centering
    \includegraphics[width=0.5\textwidth]{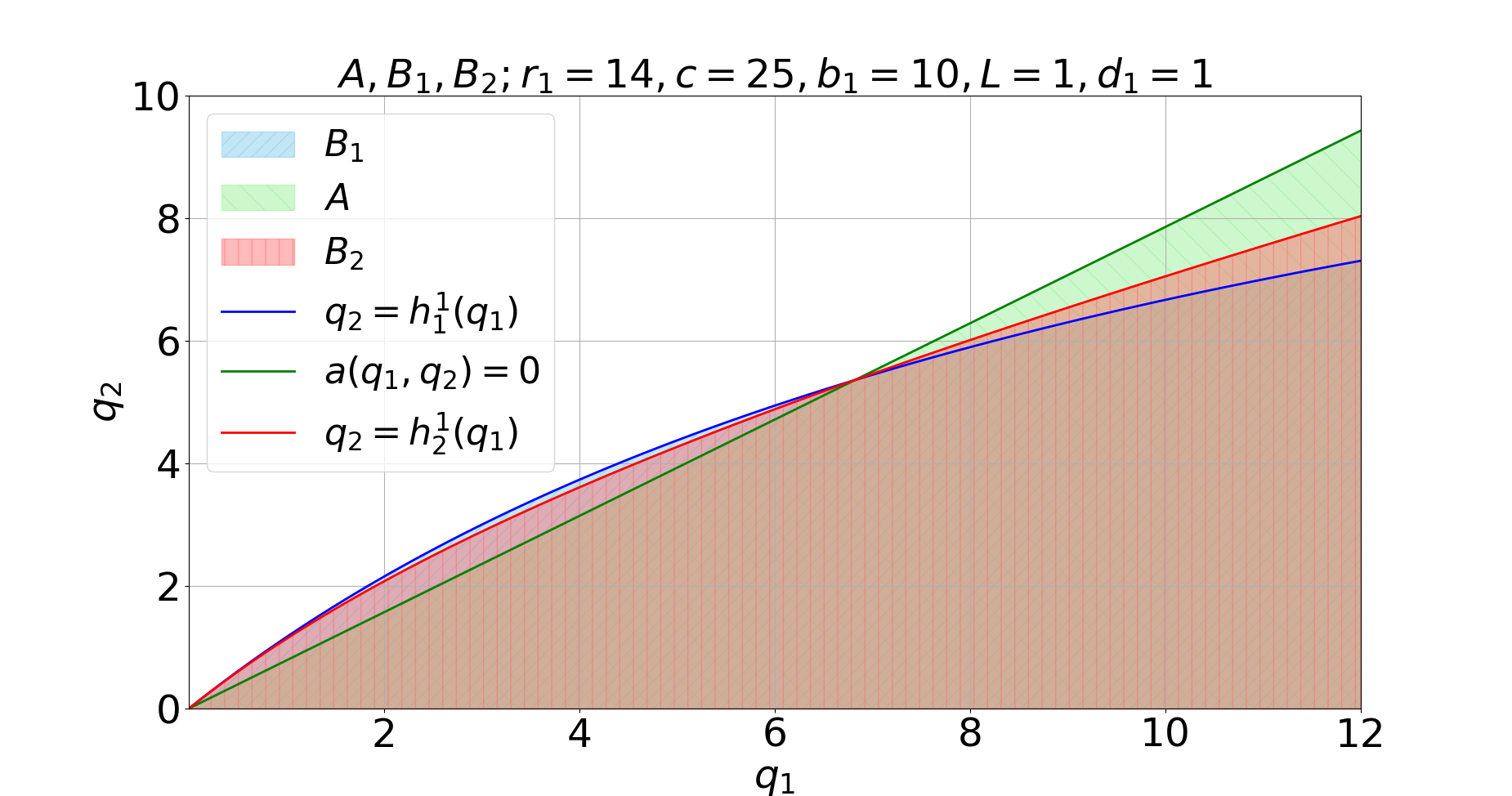}
    \caption{The shaded region in blue, red and green  shows the points that belong to the sets $B_1, B_2$ and $A$. The blue, red and green lines denote the equations $q_2 = h_1^1(q_1), q_2 = h_2^1(q_1), a(q_1, q_2) = 0$ respectively, with the parameter values given by $b_1 = 10, L = 1, C = 25, d_1 = 1$. The Brown region denotes the set $\{A \cup B_1\}$ and the small pink region denotes the set $\{A^c \cup B_2\}$.}
    \label{A,B}
\end{figure}


\begin{remark}\label{chainofimplications}
For $a(q_1, q_2) \geq 0$, the following holds:
$f_2^1(q_1, q_2) \geq f_1^1(q_1, q_2)$, $B_2 \subseteq B_1$, and $B_1 \cap B_2 = B_2$.\\
For $a(q_1, q_2) \leq 0$, the implications reverse:
$f_1^1(q_1, q_2) \geq f_2^1(q_1, q_2)$, $B_1 \subseteq B_2$, and $B_1 \cap B_2 = B_1$.
\end{remark}
We will use the observations noted in remark \ref{chainofimplications} while proving the convexity of the set $\mathcal{D}^1$

\begin{lemma}\label{lemma:notemptymeansconvex}
    The set $A^c \cap B_2$ is convex.
\end{lemma}
\begin{proof}
See Appendix~\ref{app:notemptymeansconvex}
\end{proof}

The set $A \cap B_2$ is empty if $d_1r_1 \leq b_1 + L$, a condition utilized in the next theorem.


\begin{lemma}
The set $\mathcal{D}^1$ is convex.
\end{lemma}\label{D is convex}

\begin{proof}
\textbf{Case 1:} 

Let $A^c \cap B_2 = \emptyset$. In this case, $\mathcal{D}^1 = A \cap B_1$, where the set $A$ is a half-space and is convex. 
Since $h_1^1(q_1)$ is concave, the set $\{q_2 \leq h_1^1(q_1)\}$, denoted as $B_1$, is convex. The intersection of two convex sets, $A$ and $B_1$, implies that $\mathcal{D}^1$ is convex.

\textbf{Case 2:} 

When $A^c \cap B_2 \neq \emptyset$, we know that $(A^c \cap B_2)$ is a convex set and that $h_2^1(q_1)$ is a concave function. Therefore, using the remark \ref{chainofimplications} the set $\mathcal{D}^1$ can be written as:

\begin{align*}
\mathcal{D}^1 &= ( A \cap B_1) \cup ( A^c \cap B_2) \\
                &=(A\cap B_1 \cap B_2) \cup (A^c \cap B_1 \cap B_2)\\
                &= (B_1 \cap B_2)\\
              &=  (q_2 \leq H_1(q_1)),
\end{align*}
\begin{align*}
\text{where }
H_1(q_1) &= \min \left\{ h_1^1(q_1), h_2^1(q_1) \right\}.
\end{align*}
Since both $h_1^1(q_1)$ and $h_2^1(q_1)$ are concave functions, their pointwise minimum, $H_1(q_1)$, is also concave. Therefore, the set $\mathcal{D}^1 = \left\{ q_2 \leq H_1(q_1) \right\}$ is convex.
\end{proof}

Thus, we have shown that $\mathcal{D}^1$ is a convex set. By symmetry, $\mathcal{D}^2$ is also convex. Therefore, the intersection $\mathcal{D}^1 \cap \mathcal{D}^2$ is convex.
To unify both cases, we update the definition of $H_1(q_1)$ as follows:
\begin{align}\label{H_1}
H_1(q_1) &= 
\begin{cases}
h_1^1(q_1) & \text{if} \quad d_1 r_1 \leq b_1 + L, \\
\min \left\{ h_1^1(q_1), h_2^1(q_1) \right\} & \text{if} \quad d_1 r_1 > b_1 + L.
\end{cases}
\end{align}

We now state some properties of the function $H_i(q_i)$:

\begin{remark}
    The function $H_i(q_i): (0, \infty) \to (0, \infty)$, follow these properties mentioned below,
    \begin{itemize}
        \item $\lim_{q_i \to 0}H_i(q_i) = 0$
        \item $H_i(q_i)$ is a concave, strictly increasing, continuous function.
        \item $\left.\frac{d H_i(q_i)}{d q_i}\right|_{q_1 = 0} = \frac{cd_i}{b_i+L} - 1$
        \item $\left.\frac{d H_i(q_i)}{d q_i}\right|_{q_1 = \infty} = 0$
    \end{itemize}
\end{remark}


\subsection{Convexity of feasible set: $n$ flow case}
In this subsection, we aim to demonstrate that the constraint set for the $n$-flow problem, which ensures that the delay requirements of all flows are satisfied, is a convex set. Convexity of the constraint set is essential as it allows the $n$-flow optimization problem to be framed as a convex optimization problem. This reformulation enables the use of efficient algorithms to determine the optimal solution.

\begin{theorem}
    The constraint set for the $n$-flow problem, $\bigcap_{j=1}^n \mathcal{D}^j$, is a convex set.
\end{theorem} 

\begin{proof}
The proof proceeds by reformulating the delay constraints for each flow using auxiliary variables, simplifying the problem to a two-dimensional convex constraint set. Leveraging this structure, we verify the convexity of the reformulated set. Finally, using the properties of affine transformations and the intersection of convex sets, we conclude that the original constraint set is convex.


The set $\mathcal{D}^i$ as defined in \eqref{2-flow-fesability-set} and the function $f^i(q)$ as defined in \eqref{deviation_n_flow}, we substitute $k_n^i = \sum_{j \neq i} q_j$ and $J_i = d_i - \frac{L(n-2)}{c}$ and reformulate the function $f^i(q_1, q_2, \dots, q_n)$ as:
\begin{align}\label{nto2delayequation}
    f^i_k(q_1, k_n^i) &= \frac{b_i + L}{c} \left(1 + \frac{k_n^i}{q_i}\right) 
    + \frac{k_n^i}{c} - J_i \\ 
    & + \left(q_i \left(\frac{r_i - c}{r_i c}\right) + \frac{k_n^i}{c}\right)^+. \nonumber
\end{align}
We define, $\mathcal{\hat{D}}^i = \left\{(q_1, k_n^i) \in \mathbb{R}_+^2 \mid f^i_k(q_1, k_n^i) \leq 0\right\}.$ From Theorem~\ref{D is convex}, we know that $\mathcal{\hat{D}}^i$ is a convex set. Using an affine transformation to return to the original variables, we conclude that $\mathcal{D}^i$ is also convex \cite{lay2007convex}. 

Since the intersection of convex sets is also convex, $\bigcap_{i=1}^n \mathcal{D}^i$ is a convex set.
\end{proof}

%% file: 10Minimalswitching2-flow.tex
\section{Optimization problem}\label{OP2F}
In this section, we propose two algorithms to solve the two flow and the $n$ -flow  case optimization problem as given in $\eqref{n-flow-opti}$. The Two flow case optimization problem is given by, 
\begin{equation}\label{2-flow-opti}
\begin{aligned}
   \text{maximize} \quad & q_1 + q_2 \\
   \text{subject to} \quad 
        & (q_1, q_2) \in \mathcal{D}^1 \cap \mathcal{D}^2,
\end{aligned}
\end{equation}

\subsection{Optimization: 2-flow case}
We aim to maximize the amount of data served in one round and meet the delay requirements. Since $q_1$ represents the number of packets served per round for flow 1, maximizing the sum $q_1+q_2$ increases the amount of packets served in a single round. The optimization problem defined in \eqref{2-flow-opti} is a convex optimization problem.


\begin{lemma}\label{unique-exist-2-flow}[Existence and Uniqueness of Fixed Point]
The system of equations:
\begin{align}
H_1(q_1) &= q_2, \\
H_2(q_2) &= q_1.
\end{align}
have a unique solution $(q_1^*, q_2^*)$, where $q_1^* > 0$ and $q_2^* = H_1(q_1^*)$.
\end{lemma}

\begin{proof}
    See Appendix~\ref{ap:unique-exist-2-flow}
\end{proof}
A brief outline of the proof of Theorem~\ref{unique-exist-2-flow} is as follows. Define $T(q_1) = H_2(H_1(q_1))$, a concave increasing function \cite{Boyd_Vandenberghe_2004}. First, we analyze $T'(q_1)$ and find $T'(0) > 1$, ensuring $T(q_1)$ initially grows faster than $q_1$ but eventually slower as $T'(q_1)$ decreases. Second, we establish the existence of a fixed point $q_1^*$ such that $T(q_1^*) = q_1^*$, as $T(0) = 0$ and $T(q_1) < q_1$ for large $q_1$. Finally, the concavity of $T(q_1)$ ensures the fixed point is unique by ruling out the possibility of multiple fixed points. Moreover, the fixed point $q_1^*$ represents the equilibrium of the system described by $T(q_1)$.

\begin{lemma}\label{optimal-point}
    The optimal solution to the optimization problem described in \eqref{2-flow-opti} is given by the pair of values that satisfy the following system of equations:
    \begin{align}\label{set-equation-2flow}
        H_1(q_1) &= q_2, \\
        H_2(q_2) &= q_1. \nonumber
    \end{align}
\end{lemma}

\begin{proof}
    See Appendix~\ref{app:optimal-point}
\end{proof}

In this section, we have formulated a convex optimization problem aimed at minimizing the number of switchings between two flows by maximizing the sum of the quanta $q_1 + q_2$. The existence and uniqueness of the optimal solution were established, with the solution satisfying the system of equations $H_1(q_1) = q_2$ and $H_2(q_2) = q_1$. Through analysis, we showed that the optimal solution exists and is unique, providing a set of conditions under which this solution can be determined.

\begin{theorem}\label{fixed-point-iteration}
    Consider the fixed-point iteration:
    \begin{align}
        x_{k+1} = T(x_k),
    \end{align}
    where the initial value \( x_0 > 0 \) is given. Then, the sequence \( \{x_k\} \) generated by this iteration, converges to a unique fixed point \( q_1^* > 0 \), i.e.,
    \begin{align}
        \lim_{k \to \infty} x_k = q_1^*.
    \end{align}
    Furthermore, the fixed point \( q_1^* \) satisfies the equation:
    \begin{align}
        T(q_1^*) = q_1^*.
    \end{align}
\end{theorem}

\begin{proof}
See Appendix~\ref{app:fixed-point-iteration}
\end{proof}

\begin{algorithm}
\caption{Fixed-Point Iteration for Finding $q^*$ for two-flow case}
\begin{algorithmic}[1]
    \State \textbf{Input:} Initial guess $x_0 > 0$, tolerance $\epsilon > 0$, $T(x)$
    \State \textbf{Output:} Fixed point $q^*$
    \State \textbf{Initialize:} Set $x_0$
    \State Set $k \gets 1$
    \State Set $x_1 \gets T(x_{0})$
    \While{$|x_k - x_{k-1}| > \epsilon$}
        \State Compute $x_{k+1} \gets T(x_k)$
        \State Update $k \gets k + 1$
    \EndWhile
    \State \textbf{Return} $q_1^* = x_k$; $q_2^* = H_1(q_1^*)$
\end{algorithmic}\label{algo-two-flow-case}
\end{algorithm}

In Figure~\ref{two_flow_fixed_point}, the trajectory of $q_1$, as determined by the  Algorithm~\ref{algo-two-flow-case}, is shown to converge to the optimal value $q_1^*$. The specified parameters for this case are $b=[10,15], r=[1,2], d=[1,0.5], c=50,$ and $L=1$. Starting with an initial guess of $x_0=1$, the algorithm successfully converges to the optimal solution, $(q_1^*, q_2^*) = (3.181, 8.749)$.

\begin{figure}[htbp!]
    \centering
    \includegraphics[width=0.5\textwidth]{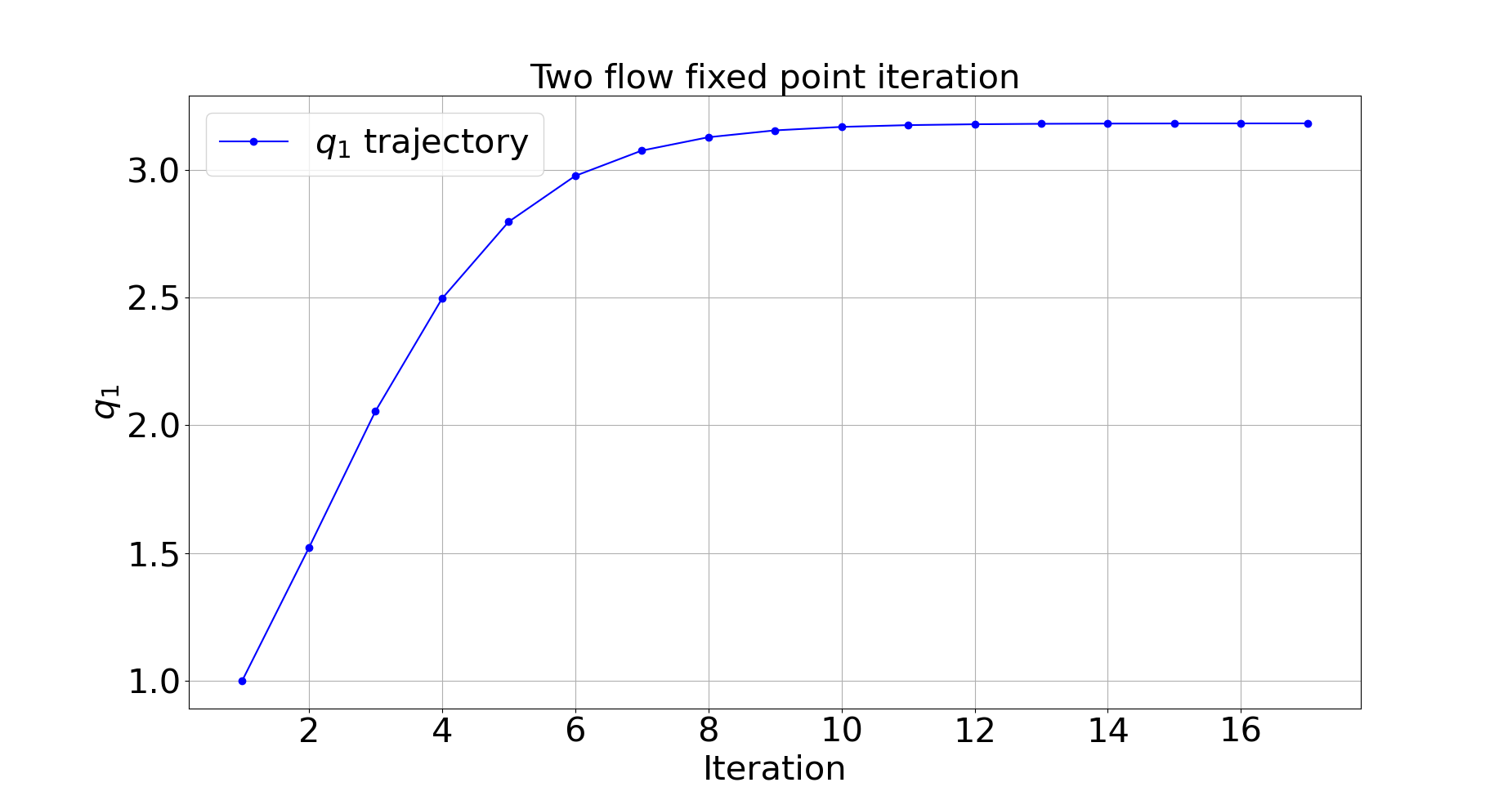}
    \caption{Trajectory of $q_1$ under Algorithm~\ref{algo-two-flow-case}, illustrating convergence to the optimal value $q_1^* = 3.181$. The parameters for this case are $b = [10, 15]$, $r = [1, 2]$, $d = [1, 0.5]$, $c = 50$, and $L = 1$, with an initial guess of $x_0 = 1$.}

    \label{two_flow_fixed_point}
\end{figure}

\subsection{Optimization: n-flow case}
We now extend the problem to the case of minimizing the number of switchings in the $n$-flow scenario, where the goal is to maximize the total quanta while satisfying the flow constraints for each of the $n$ flows. The optimization problem is defined as follows:

\begin{equation}
\begin{aligned}
   \text{maximize} \quad & \sum_{i=1}^{n} q_i \\
   \text{subject to} \quad 
        & (q_1, q_2, \dots, q_n) \in \bigcap_{j=1}^{n} \mathcal{D}^j,
\end{aligned}
\end{equation}

where $q_1, q_2, \dots, q_n$ represent the quanta for each flow, and the feasible region is given by the intersection of the constraint sets $\mathcal{D}^j$ for each $j = 1, 2, \dots, n$. These constraint sets represent the feasible values for the quanta of each individual flow, ensuring the system operates within the required delay requirements, such as delay and throughput constraints for each flow. The objective of this problem is to maximize the sum of all quanta, which corresponds to minimizing the number of switching between the flows.


Since this is a convex optimization problem, it can be solved using any standard convex optimization solver. These solvers iteratively explore the feasible space, adjusting the quanta to maximize the objective function. The convexity of the problem guarantees that any local optimum is also a global optimum, ensuring that the solution found by such solvers is optimal for the given system.

We propose a fixed-point iterative method to find the optimal solution, as described below.

From \eqref{nto2delayequation}, the constraint set can be expressed as \begin{align} \sum_{j \neq i} q_j \leq H_i(q_i). \end{align} Since the function $H_i(q_i)$ is strictly increasing and concave, we define
\begin{align} g_i(q_i) = H_i(q_i) + q_i. \end{align} Thus, the constraint can be rewritten as
\begin{align} \sum_{j} q_j \leq g_i(q_i). \end{align} Since $g_i(q_i)$ is increasing and concave, let us assume $\sum_i q_i = \theta$. This leads to
\begin{align} \sum_i g_i^{-1}(\theta) \leq \theta. \end{align}

\begin{lemma}\label{fixed_point_n}
    There exists an unique $\theta^* > 0$, such that
    \begin{align}
        \sum_i g_i^{-1}(\theta^*) = \theta^*
    \end{align}
\end{lemma}

\begin{proof}
    Let, $\Gamma(\theta) = \sum_i g_i^{-1}(\theta)$. We know that $\Gamma(\theta)$ is an increasing and convex function. Also, we know that
    \begin{align}\label{gamma_limit_0}
        \lim_{\theta \to 0}\Gamma^{'}(\theta) &= \frac{1}{c} \sum_j \frac{b_j + L}{d_j} \leq 1\\
        \lim_{\theta \to \infty}\Gamma^{'}(\theta) &= n > 1
    \end{align}
The inequality in \eqref{gamma_limit_0} is from the non-empty condition given in Lemma~\ref{n_flow_non_empty_lemma}. By intermediate value theorem, there exists at least one $\theta > 0$, such that $\Gamma(\theta) = \theta$. 
Let there exist two distinct fixed points $\theta_1$ and $\theta_2$ such that $\Gamma(\theta_1) = \theta_1$ and $T_2(\theta_2) = \theta_2$,
WLOG we assume $\theta_1 < \theta_2$. $T(\theta)$ is strictly increasing $T(\theta_1) < T(\theta_2)$. The line connecting $(T(\theta_1), \theta_1)$ and $(T(\theta_2), \theta_2)$ is $1$. By, mean value theorem there exist $c \in (\theta_1, \theta_2)$ such that $T^{'}(c) = 1$ \cite{rockafellar1970convex}. But, $T(\theta)$ is concave, $T(\theta)^{'}$ is increasing and for any $\theta < c$ the slope is less than $1$, and it cannot intersect the line y = x at $\theta_1$. This contradicts the assumption that $\theta_2$ is a fixed point. 
\end{proof}

\begin{lemma}\label{optimal-point}
    The optimal solution to the optimization problem in \eqref{n-flow-opti} is given by \( q^* \), where each component satisfies
    \[
    q_i^* = g_i^{-1}(\theta^*).
    \]
    The parameter \( \theta^* \) is uniquely determined by the fixed-point equation:
    \begin{align}\label{set-equation-2flow}
        \Gamma(\theta^*) = \theta^*.
    \end{align}
\end{lemma}
\begin{proof}
    The feasibility condition for the optimization problem \eqref{n-flow-opti} is:
    \begin{align}
        \sum_i g_i^{-1}(\theta) &\leq \theta,\\
        \Gamma(\theta ) &\leq \theta.
    \end{align}
    Since \( \Gamma(\theta) \) is an increasing function, Theorem~\ref{fixed_point_n} guarantees the existence of a unique \( \theta^* > 0 \) satisfying \( \Gamma(\theta^*) = \theta^* \). 

    Furthermore, \( \Gamma(\theta) \) is convex, and at \( \theta^* \), we have \( \Gamma'(\theta^*) \geq 1 \). By the first-order convexity condition, for any \( \theta > \theta^* \),
    \[
    \Gamma(\theta) \geq \Gamma(\theta^*) + \Gamma'(\theta^*) (\theta - \theta^*) \geq \theta.
    \]
    This implies that values \( \theta > \theta^* \) are not feasible. Similarly, for \( \theta < \theta^* \), by definition, we obtain
    \[
    \sum_j q_j < \sum_j q_j^*.
    \]
    Hence, the optimal solution to the problem in \eqref{n-flow-opti} is given by \( q_i^* = g_i^{-1}(\theta^*) \).
\end{proof}

\begin{theorem}
    Consider the fixed-point iteration:
    \begin{align}
        \theta_{k+1} = \Gamma^{-1}(\theta_k)
    \end{align}
    where the initial value $\theta_0 > 0$ is given. Then, the sequence $\{\theta_k\}$ generated by this iteration converges to a unique fixed point $\theta^*$, i.e.,
    \begin{align}
        \lim_{k \to \infty} \theta_k = \theta^*
    \end{align}
    Furthermore, the fixed point $\theta^*$ satisfies the equation 
    \begin{align}
        \Gamma(q_1^*) = q_1^*
    \end{align}
\end{theorem}

\begin{proof}
   We know that $\Gamma(\theta)$ is an increasing and convex function. Therefore, $\Gamma^{-1}(\theta)$ is an increasing and concave function. The properties  are similar to the function $T(x)$, hence the proof is same as theorem~\ref{fixed-point-iteration}. 
\end{proof}
These algorithms can be viewed as a fixed-point iteration connected to the nonlinear Perron-Frobenius theory, establishing links to well-established mathematical frameworks presented in \cite{Zheng-etal-pft-16, Zheng-etal-max-min-2018, Gao-opti-delay-2016}

\begin{algorithm}
\caption{Fixed-Point Iteration for Finding $q^*$ for $n$-flow case}
\begin{algorithmic}[1]
    \State \textbf{Input:} Initial guess $\theta_0 > 0$, tolerance $\epsilon > 0$, $g_i^{-1}(\theta), \Gamma^{-1}(\theta)$
    \State \textbf{Output:} Fixed point $q^*$
    \State \textbf{Initialize:} Set $\theta_0$
    \State Set $k \gets 1$
    \State Set $\theta_1 \gets \Gamma^{-1}(\theta_{0})$
    \While{$|x_k - x_{k-1}| > \epsilon$}
        \State Compute $\theta_{k+1} \gets \Gamma^{-1}(\theta_k)$
        \State Update $k \gets k + 1$
    \EndWhile
    \State \textbf{Return} $q_i^* = g_i^{-1}(\theta_k)$
\end{algorithmic}\label{algo-n-flow-case}
\end{algorithm}

\begin{figure}[htbp!]
    \centering
    \includegraphics[width=0.54\textwidth]{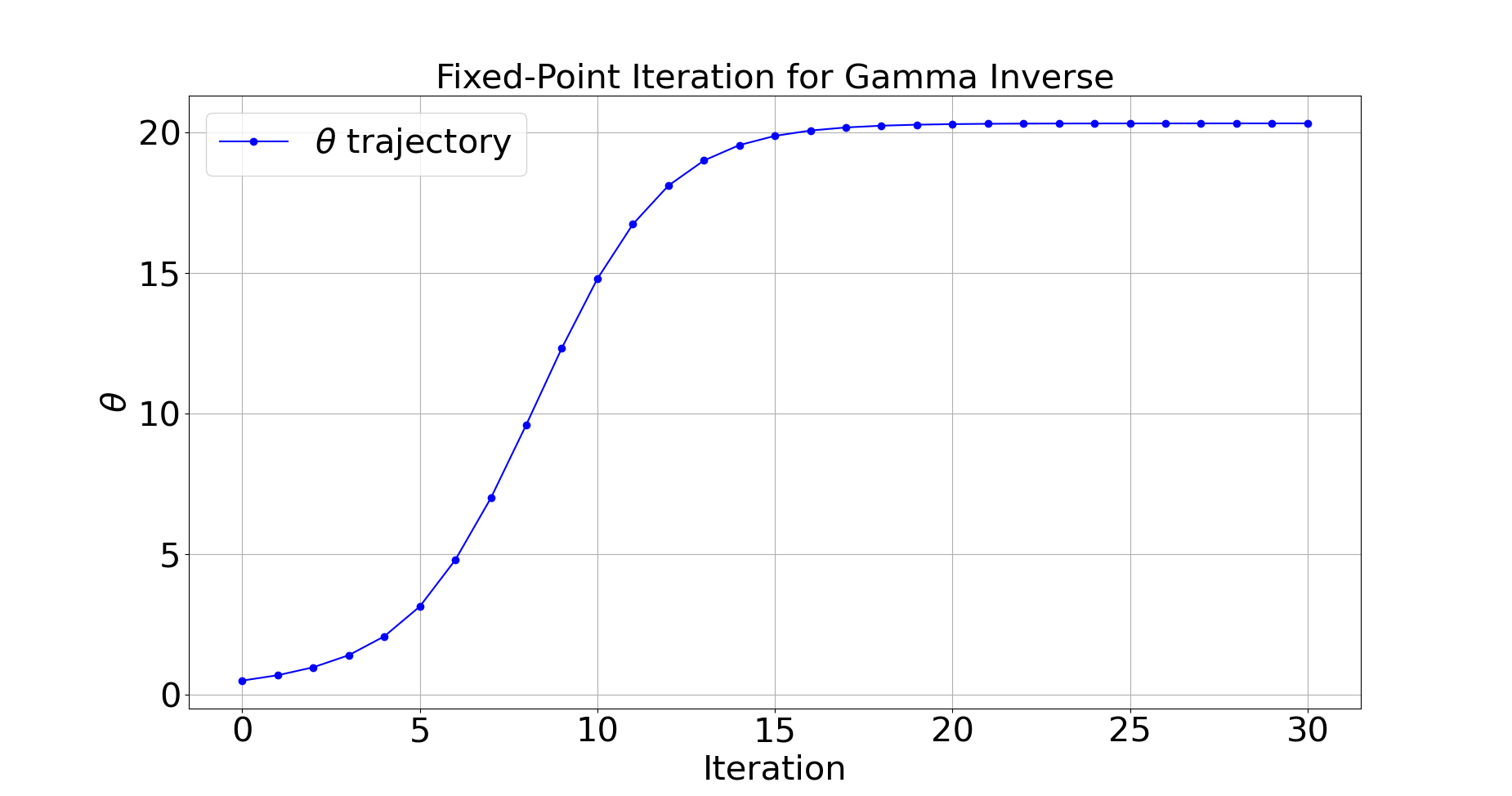}
    \caption{Trajectory of $q_1$ under Algorithm~\ref{algo-two-flow-case}, illustrating convergence to the optimal value $\theta^* = 20.312$. The parameters for this case are $b = [10, 15, 10]$, $r = [1, 2, 1]$, $d = [1, 1, 0.5]$, $c = 100$, and $L = 1$, with an initial guess of $\theta_0 = 0.5$.}

    \label{n_flow_fixed_point}
\end{figure}

In Figure~\ref{n_flow_fixed_point}, the trajectory of $\theta$, as determined by the  Algorithm~\ref{algo-n-flow-case}, is shown to converge to the optimal value $\theta^*$. The specified parameters for this case are $b=[10,15,10], r=[1,2,1], d=[1,1,0.5], c=100,$ and $L=1$. Starting with an initial guess of $\theta_0=0.5$, the algorithm successfully converges to the optimal solution, $\theta^* = 20.312$.

\begin{table}[htbp!]
\centering
\begin{tabular}{|p{0.5cm}|p{0.45cm}|p{0.72cm}|p{0.55cm}|p{0.71cm}|p{0.77cm}|p{0.76cm}|p{0.95cm}|}
\hline
No. Flows & Burst (Kb) & Packet Length (Kb) & Rate (Kb/s) & Delay Target (ms) & Capacity (Kb/s) & Optimal Quanta (b) & Simulated Delay (ms) \\
\hline
\multirow{2}{*}{2} & 10 & \multirow{2}{*}{3} & 1 & 1 & \multirow{2}{*}{40} & 14000 & 0.2 \\
& 10 &  & 1 & 1 &  & 14000 & 0.2 \\
\hline
\multirow{2}{*}{2} & 10 & \multirow{2}{*}{5} & 1 & 1 & \multirow{2}{*}{80} & 11764 & 0.5 \\
& 15 &  & 2 & 0.5 &  & 28571 & 0.1 \\
\hline
 \multirow{3}{*}{3} & 10 & \multirow{3}{*}{5} & 1 & 1 & \multirow{3}{*}{150} & 1144 & 0.5 \\
& 15 &  & 2 & 0.5 &  & 4211 & 0.2 \\
& 15 &  & 2 & 0.5 &  & 4211 & 0.2 \\
\hline
\multirow{4}{*}{4} & 10 & \multirow{4}{*}{5} & 1 & 1 & \multirow{4}{*}{200} & 2641 & 0.6 \\
& 15 &  & 2 & 0.5 &  & 10593 & 0.2 \\
& 15 &  & 2 & 0.5 &  & 10593 & 0.2 \\
& 10 &  & 1 & 0.75 &  & 3872 & 0.4 \\
\hline

\end{tabular}
\caption{Quanta values from Algorithms~\ref{algo-two-flow-case} and~\ref{algo-n-flow-case} compared to DRR simulation delays.}
\label{Table_delay}
\end{table}

We simulate the DRR scheduler using the optimal quanta obtained from the convex optimization solution and observe the actual delays. Since the optimization provides real-valued quanta, we convert them to integer values by taking the floor function. This conversion does not guarantee that the delay requirement will be met exactly. However, our simulations show that the delay constraints are still satisfied, as shown in Table~\ref{Table_delay}. This is because the optimization approach was conservative, using an upper bound in the delay calculation, which leads to an over-designed solution. As a result, even after rounding down the quanta, the observed delays remain within acceptable limits.

%% file: 11nflow-case.tex

%% file: 12results.tex

%% file: 13conclusion.tex
\section{conclusion}\label{conclusion}
This paper presents an optimization framework designed to maximize the number of packets served in a round while optimizing total quanta and meeting delay requirements. We have formulated a convex optimization problem, derived the necessary conditions for the optimal solution, proposed two algorithms to find the optimal solution and proved the existence and uniqueness of a fixed point for the governing system of equations. Our theoretical results provide a solid understanding of the interaction between the two flows and how to maximize the total number of packets served in a round of DRR scheduler.
In the practical simulation, we mapped real-valued quanta to integer-valued quanta by applying the floor function, as integer quanta are commonly required in network systems. The simulation demonstrated that the delay requirements were successfully met. However, it is important to note that the simulation does not guarantee the delay constraints will always be satisfied in all cases. The use of the floor function introduces an approximation, which may not always hold, resulting in some uncertainty. Ensuring that the delay requirements are consistently met despite this approximation, remains an open problem. In future work, we plan to explore methods for establishing more robust guarantees to formally ensure delay constraints are satisfied. Additionally, we aim to refine our approach to manage better integer-valued constraints without compromising system performance and efficiency.

%% file: 14Appendix.tex
\appendix
\subsection{ proof of Lemma~\ref{modifieddelaybound}}\label{app:lemma_proof}

\begin{proof}
Substituting the expressions for $\psi_i(b_i)$ and $\psi_i(\alpha_i(\tau_i))$ from \eqref{psi}, \eqref{phi} in the delay bound $D_i$ we get,

\begin{align}
    D_i &= \frac{b_i}{c} + \sum_{j \neq i}\left(\frac{q_j + L}{c}\right)+ \max\left\{\sum_{j \neq i}\left\lfloor \frac{b_i+L}{q_i}\right\rfloor\frac{q_j}{c}, \right.\\
    &\left. q_i - (b_i+L) \text{mod} q_i \left(\frac{1}{c}-\frac{1}{r_1} \right) + \sum_{j \neq i} \right.\nonumber\\
    &\left.\left(\left\lfloor \frac{\alpha_i(\tau_i) + L}{q_i}\right\rfloor \frac{q_j}{c}\right) \right\} \nonumber
\end{align}

We know that $\alpha_i(\tau_i) = b_i + q_i - (b_i + L) \text{mod} q_i$, 
Substituting the value of $\alpha_i(\tau_i)$ and removing the floor we get

\begin{align}
    D_i &\leq \frac{b_i + L}{c}\left(1 + \frac{\sum_{j \neq i} q_j}{q_i} \right) + \frac{\sum_{j \neq i}q_j}{c} + \frac{(n-2)L}{c} \\
    &(q_i-(b_i+L) \text{mod} q_i)\left(\frac{1}{c} - \frac{1}{r_i} + \frac{\sum_{j \neq i} q_j}{q_i c}\right) ^+ \nonumber
\end{align}

Now, we know that $0 \leq q_i - (b_i+L) \text{mod} q_i \leq q_i$, therefore we replace $q_i - (b_i+L) \text{mod} q_i$ with $q_i$ to get an upper bound,

\begin{align}
    D_i &\leq \frac{b_i + L}{c}\left(1 + \frac{\sum_{j \neq i} q_j}{q_i} \right) + \frac{\sum_{j \neq i}q_j}{c} + \frac{(n-2)L}{c} \nonumber \\
    &\left(q_i\left(\frac{1}{c} - \frac{1}{r_i} + \frac{\sum_{j \neq i} q_j}{q_i c}\right) \right)^+ \label{modifiedboundinteger}
\end{align}


\end{proof}

\subsection{Proof of Lemma~\ref{lemma:notemptymeansconvex}}\label{app:notemptymeansconvex}

\begin{proof}
If $h_2^1(q_1)$ is concave, then the set  $\{q_2 \leq h_2^1(q_1)\} = B_2$ is convex. 
The condition for $h_2^1(q_1)$ to be concave is given by:

\begin{align}
    \frac{2 c r_1 d_1}{r_1 + c} > b_1 + L. \tag{23}
\end{align}
If the function $h_2^1(q_1)$ is concave, then the set $(A^c \cap B_2)$ is a convex set.

Next, consider the slope of the half-space boundary $A^c$, which is given by:
\begin{align}
    \frac{c}{r_1} - 1. \tag{24}
\end{align}
The slope of $h_2^1(q_1)$ at $q_1 = 0$ is given by:
\begin{align}
    \frac{c d_1}{b_1 + L} - 1. \tag{25}
\end{align}
The slope of $h_2^1(q_1)$ as $q_1 \to \infty$ is given by, $\frac{c}{r_1} - 1$ and the function $h_2^1(q_1)$, The set $(A^c \cap B_2)$ is non-empty if the slope of $h_2^1(q_1)$ at $q_1 = 0$ is greater than the slope of the boundary of the half-space $A$. This condition can be expressed as:
\begin{align}
    d_1 r_1 > b_1 + L. \tag{26}
\end{align}
Furthermore, we observe that:
\begin{align}
    \frac{2 c r_1 d_1}{r_1 + c} > d_1 r_1 > b_1 + L. \tag{27}
\end{align}
Therefore, if the set $(A^c \cap B_2) \neq \emptyset$, then $(A^c \cap B_2)$ is convex.
\end{proof}

\subsection{proof of Lemma~\ref{unique-exist-2-flow}}\label{ap:unique-exist-2-flow}
\begin{proof}
Let us define,
\begin{align}
    T(q_1) = H_2(H_1(q_1))
\end{align}
We know that $H_1$ and $H_2$ are concave functions and are increasing, hence the function $T$ is also a concave function \cite{Boyd_Vandenberghe_2004}. We proceed in three steps: analyzing the behavior of $T'(q_1)$, proving existence, and proving uniqueness.

\textbf{Step 1: Behavior of $T'(q_1)$.}

From the assumptions, the derivative of $T$ is given by:
\begin{align}
T'(q_1) = H_2'(H_1(q_1)) H_1'(q_1).
\end{align}

At $q_1 = 0$, we are given:
\begin{align}
T'(0) = \left( \frac{cd_2}{b_2+L} - 1 \right)\left( \frac{cd_1}{b_1+L}  - 1 \right).
\end{align}

For the existence of a fixed point, we require that $T'(0) > 1$. This condition ensures that $T$ initially grows faster than the line $q_1$, which is critical for an intersection.

As $q_1 \to \infty$, concavity of $H_1$ and $H_2$ implies that the slopes decrease. Specifically, we are given:
\begin{align}
\lim_{q_1 \to \infty} T'(q_1) = 0.
\end{align}

Thus, $T(q_1)$ grows slower than $ q_1$ for large values of $q_1$.

\textbf{Step 2: Existence of a Fixed Point.}

At $ q_1 = 0$, we know $T(0) = H_2(H_1(0)) = 0$. For large $q_1$, since $T'(q_1) \to 0 $, we have:
\begin{align}
\lim_{q_1 \to \infty} T(q_1) < q_1.
\end{align}

By continuity of $ T $, there must exist a point $q_1^* > 0$ such that:
\begin{align}
T(q_1^*) = q_1^*.
\end{align}
This follows from the Intermediate Value Theorem, as $T(0) = 0$ and $ T(q_1) < q_1 $ for sufficiently large $ q_1 $.

\textbf{Step 3: Uniqueness of the Fixed Point.}
Let there exist two distinct fixed points $x_1$ and $x_2$ such that $T(x_1) = x_1$ and $T_2(x_2) = x_2$,
WLOG we assume $x_1 < x_2$. $T(x)$ is strictly increasing $T(x_1) < T(x_2)$. The line connecting $(T(x_1), x_1)$ and $(T(x_2), x_2)$ is $1$. By, mean value theorem there exist $c \in (x_1, x_2)$ such that $T^{'}(c) = 1$ \cite{rockafellar1970convex}. But, $T(x)$ is concave, $T(x)^{'}$ is decreasing and for any $x > c$ the slope is less than $1$, and it cannot intersect the line y = x at $x_2$. This contradicts the assumption that $x_2$ is a fixed point. 

\textbf{Step 4: Corresponding \( q_2^* \).}

From the definition \( T(q_1^*) = H_2(H_1(q_1^*)) = q_1^* \), we let:
\begin{align}
q_2^* = H_1(q_1^*).
\end{align}
\end{proof}

\subsection{Proof of Theorem~\ref{optimal-point}}\label{app:optimal-point}

\begin{proof}
From Theorem \ref{unique-exist-2-flow}, let the unique point that satisfies the set of equations \eqref{set-equation-2flow} be $q^* = (q_1^*, q_2^*)$. Assume $\hat{q} = (\hat{q_1}, \hat{q_2})$ is an optimal solution such that $\hat{q_1} + \hat{q_2} > q_1^* + q_2^*$. Without loss of generality (WLOG), assume $\hat{q_1} > q_1^*$.

From the feasibility condition and the fact that the set of equations $H_1(q_1) = q_2$ and $H_2(q_2) = q_1$ has a unique solution, we assume, under the assumption that $\hat{q_2} \leq H_1(\hat{q_1})$, that $\hat{q_1} < H_2(\hat{q_2})$:
\begin{align}
    \hat{q_2} \leq H_1(\hat{q_1}) &= H_1(q_1^*) + \left(H_1(\hat{q_1}) - H_1(q_1^*)\right).
\end{align}
Using the concavity of $H_1$ and the fact that $H_1(0) = 0$, we can write:
\begin{align}
    H_1'(q_1^*) &\leq \frac{H_1(q_1^*)}{q_1^*} = \frac{q_2^*}{q_1^*}.
\end{align}
By concavity again, we have:
\begin{align}
    H_1(\hat{q_1}) - H_1(q_1^*) &\leq H_1'(q_1^*) (\hat{q_1} - q_1^*),
\end{align}
which implies:
\begin{align}
    \hat{q_2} &\leq H_1(q_1^*) + \frac{q_2^*}{q_1^*} (\hat{q_1} - q_1^*) = \frac{q_2^* \hat{q_1}}{q_1^*}.
\end{align}

Thus, we have:
\begin{equation}\label{optimality-condition-1}
    \frac{\hat{q_2}}{\hat{q_1}} \leq \frac{q_2^*}{q_1^*}.
\end{equation}

\textbf{Case 1:} Suppose $\hat{q_2} \geq q_2^*$.  
\begin{align}
    \hat{q_1} < H_2(\hat{q_2}) &= H_2(q_2^*) + \left(H_2(\hat{q_2}) - H_2(q_2^*)\right).
\end{align}
We assume that $\hat{q_1} < H_2(\hat{q_2})$, because we know that the system of equations $H_1(q_1) = q_2$ and $H_2(q_2) = q_1$ has exactly one unique positive solution. Additionally, we have already assumed that $\hat{q_2} \leq H_1(\hat{q_1})$, which justifies the assumption $\hat{q_1} < H_2(\hat{q_2})$.

Using the concavity of $H_2$ and the fact that $\hat{q_2} > q_2^*$, we have:
\begin{align}
    H_2(\hat{q_2}) - H_2(q_2^*) &\leq H_2'(q_2^*) (\hat{q_2} - q_2^*),
\end{align}
where $H_2'(q_2^*) \leq \frac{q_1^*}{q_2^*}$. Therefore:
\begin{align}
    \hat{q_1} &< H_2(q_2^*) + \frac{q_1^*}{q_2^*} (\hat{q_2} - q_2^*) = \frac{q_1^* \hat{q_2}}{q_2^*}.
\end{align}
This gives:
\begin{equation}\label{optimality-condition-2}
    \frac{\hat{q_2}}{\hat{q_1}} > \frac{q_2^*}{q_1^*}.
\end{equation}

Combining inequalities \eqref{optimality-condition-1} and \eqref{optimality-condition-2}, we have a contradiction.

\textbf{Case 2:} Suppose $\hat{q_2} \leq q_2^*$.  
From the definition of $q^*$ and the monotonicity of $H_2$, we have:
\begin{align}
    \hat{q_1} > q_1^* = H_2(q_2^*) \geq H_2(\hat{q_2}) \geq \hat{q_1}.
\end{align}
This leads to the contradiction $\hat{q_1} > \hat{q_1}$, which is impossible.  

Therefore, Case 2 cannot occur.  

we conclude that $q^*$ is the optimal solution to the optimization problem \eqref{2-flow-opti}.
\end{proof}

\subsection{ proof of Theorem~\ref{fixed-point-iteration}}\label{app:fixed-point-iteration}
\begin{proof}
    We know that the function $T(q_1)$ is concave, increasing, and has a slope greater than 1 at $q_1 = 0$.
    Let $q_1^* > 0$ be such that $T(q_1^*) = q_1^*$.

    \textbf{Case 1:} Suppose $0 < x_0 < q_1^*$. Then, we know that
    \begin{align}
        x_k < T(x_k) = x_{k+1}.
    \end{align}
    This implies that $\{x_k\}$ is an increasing sequence, and since $q_1^*$ is an upper bound, we need to show that
    \begin{align}
        \lim_{k \to \infty} x_k = \sup \{x_k\} = q_1^*.
    \end{align}
    Suppose there exists some $x' < q_1^*$ such that $x' = \sup \{x_k\}$. Then, for any $x \in (x', q_1^*)$, we have
    \begin{align}
        x' < x < T(x') < T(x).
    \end{align}
    This contradicts the assumption that $x'$ is the least upper bound of $\{x_k\}$, implying that $q_1^*$ is the limit point of the sequence.

    \textbf{Case 2:} If $x_0 > q_1^*$, then a similar argument holds, but now $\{x_k\}$ forms a decreasing sequence lower-bounded by $q_1^*$. Thus, the sequence converges to $q_1^*$.

    Hence, the fixed-point iteration converges to a unique fixed point $q_1^* > 0$ that satisfies $T(q_1^*) = q_1^*$.
    \end{proof}
    


